\documentclass{IEEEtran}
\usepackage{graphicx}
\usepackage{amsmath, amssymb, bm,bbm}
\usepackage{amsthm,cases}
\usepackage{subfigure}
\usepackage{epstopdf}
\usepackage{cite}
\usepackage{latexsym}
\usepackage{epstopdf}
\usepackage{arydshln}
\usepackage{amsfonts}
\usepackage{accents}
\usepackage{multirow}
\usepackage{dsfont}
\usepackage{enumitem}

\newtheorem{definition}{Definition}
\newtheorem{theorem}{Theorem}

  \makeatletter
  \def\widebar{\accentset{{\cc@style\underline{\mskip10mu}}}}
  \makeatother

\newcommand{\la}{\lambda}

\begin{document}

\title{$M$-Channel Critically Sampled Spectral Graph Filter Banks With Symmetric Structure}%

\author{Akie~Sakiyama,
Kana~Watanabe, and
Yuichi~Tanaka
\thanks{This work was supported in part by JST PRESTO Grant Number JPMJPR1656.}
\thanks{The authors are with the Graduate School of BASE, Tokyo University of Agriculture and Technology, Koganei, Tokyo, 184-8588 Japan. Y. Tanaka is also with PRESTO, Japan Science and Technology Agency, Kawaguchi, Saitama, 332-0012, Japan (email: sakiyama@msp-lab.org; nbkn@msp-lab.org; ytnk@cc.tuat.ac.jp).}
}

\markboth{}{}
\maketitle

\begin{abstract}
This paper proposes a class of $M$-channel spectral graph filter banks with a symmetric structure, that is, the transform has sampling operations and spectral graph filters on both the analysis and synthesis sides. The filter banks achieve maximum decimation, perfect recovery, and orthogonality.  Conventional spectral graph transforms with decimation have significant limitations with regard to the number of channels, the structures of the underlying graph, and their filter design. The proposed transform uses sampling in the graph frequency domain. This enables us to use any variation operators and apply the transforms to arbitrary graphs even when the filter banks have symmetric structures. We clarify the perfect reconstruction conditions and show design examples. An experiment on graph signal denoising conducted to examine the performance of the proposed filter bank is described.
\end{abstract}
\begin{IEEEkeywords}
Graph signal processing, spectral graph filter bank, critically sampled design, spectral domain sampling
\end{IEEEkeywords}
%

\section{Introduction}
\label{sec:intro}

\subsection{Motivation}
Signal processing on graphs has been developed to allow traditional signal processing techniques to be utilized on data with irregular complex structures. The data are processed while considering their structures, defined through graphs, and can be applied to many practical applications, such as social networks\cite{Sandry2013}, images/videos \cite{Hu2013, Cheung2018}, traffic\cite{Crovel2003}, and sensor networks\cite{Sakiya2016,Sakiya2017}.

Wavelets and filter banks with decimation are important techniques for processing or analyzing graph signals\cite{Sakiya2016b}, as well as regular signals\cite{Gan2003,Tran1997,Tanaka2004,Tanaka2009}. Although there are some approaches for $M$-channel graph filter banks with (maximum) decimation\cite{Teke2017,Jin2017,Pei2017}, they have a significant limitation in terms of the graph structure or filter design, or require a complex interpolation on the synthesis side for ensuring a perfect reconstruction. Popular approaches include critically sampled spectral graph wavelets\cite{Narang2012,Narang2013} and $M$-channel oversampled spectral graph filter banks\cite{Tanaka2014a,Sakiya2014a}. Although such techniques have relatively fewer restrictions than other conventional methods, they are applicable solely to the bipartite graphs and can only use normalized graph Laplacians as a variation operator.

All of these approaches utilize sampling operations in the vertex domain, which selects maintained vertices and samples according to the graph structure or applications. The vertex domain sampling corresponds to the sampling in the time/spatial domain for regular signals. In traditional signal processing, sampling in the time/spatial domain can also be represented in the frequency domain, and both sampling results are the same \cite{Vaidya1993}. Unfortunately, graph signal processing does not inherit this property, i.e., relationships between the spectra of the original and downsampled signals are not clarified for the vertex domain sampling, except for a bipartite case. Hence, a downsampled signal occasionally has a significantly different spectrum from that of the original signal.

This paper proposes the use of $M$-channel critically sampled spectral graph filter banks with a symmetric structure. This is a generalized version of our work \cite{Watana2018,Sakiya2018} from two- to $M$-channel structures. It uses down- and upsampling operations in the graph frequency domain\cite{Tanaka2017}, and therefore, the decomposed signal keeps the spectral information of the original graph signal. The proposed filter banks can be applied to any graph signals regardless of the structure of the underlying graph and can use any variation operators, such as a combinatorial/normalized graph Laplacian and graph adjacency matrix. We clarify the perfect reconstruction conditions, which are fortunately similar to those of regular signals. We also demonstrate that any filter sets of $M$-channel real-valued linear phase filter banks for regular signals can also be used to that for graph signals with perfect recovery. We provide filter design examples and apply the proposed method to the denoising of a graph signal.

\subsection{Notation}
\label{sec:notations}
A graph $\mathcal{G}$ is represented as $\mathcal{G} = (\mathcal{V}, \mathcal{E})$, where $\mathcal{V}$ and $\mathcal{E}$ denote sets of vertices and edges, respectively. The $(m,n)$-th element of an adjacency matrix $\mathbf{A}$ is $a_{mn} > 0$ if the $m$th and $n$th vertices are connected, or zero otherwise, where $a_{mn}$ denotes the weight of the edge between $m$ and $n$. The degree matrix $\mathbf{D}$ is a diagonal matrix, and its $m$th diagonal element is $d_{mm} = \sum_n a_{mn}$. Combinatorial and symmetric normalized graph Laplacians are defined as $\mathbf{L}:=\mathbf{D}-\mathbf{A}$ and $\bm{\mathcal{L}} := \mathbf{D}^{-1/2}\mathbf{L}\mathbf{D}^{-1/2}$, respectively. Because $\mathbf{L}$ (or $\bm{\mathcal{L}}$) is a real symmetric matrix, $\mathbf{L}$ is always decomposed into $\mathbf{L} = \mathbf{U} \bm{\Lambda} \mathbf{U}^\top$, where $\mathbf{U} = [{\bm u}_{0}, \ldots, {\bm u}_{N-1}]$ is an orthonormal eigenvector matrix; $\bm{\Lambda} = \text{diag}(\la_0, \la_1, \ldots, \la_{N-1})=\text{diag}({\bf \Lambda}_0,{\bf \Lambda}_1,\ldots,{\bf \Lambda}_{M-1})$ is a diagonal eigenvalue matrix in which the eigenvalue is the \textit{graph frequency}, ${\bf \Lambda}_i=\text{diag}(\la_{Ni/M}, \ldots, \la_{N(i+1)/M-1})$; and $\cdot^\top$ represents the transpose of a matrix or vector. For simplicity, we assume that the eigenvalues $\lambda_i$ have the following order: 
\begin{equation}0=\lambda_0<\lambda_1\leq\lambda_2\leq\cdots\leq\lambda_{N-1}=\lambda_\text{max}.
\end{equation}

Here, $\bm{f}\in \mathbb{R}^N$ is a graph signal, where the $n$th sample $f[n]$ is assumed to be located on the $n$th vertex of the graph.
The graph Fourier transform (GFT) is defined as $\widetilde{\bm f}={\bf U}^\top{\bm f}$,
and the inverse GFT is ${\bm f}={\bf U}\widetilde{\bm f}$. The graph spectral filtering can be represented as ${\bm f}_\text{out} = {\bf U}H({\bf \Lambda}){\bf U}^\top{\bm f}_\text{in}$ with a filter kernel, $H(\cdot)$.

\begin{figure}[tp]%
\centering
\includegraphics[width = .9\linewidth]{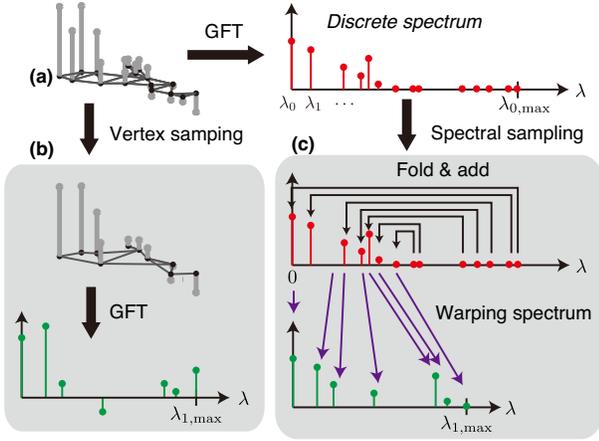}\\
\caption{Downsampling of graph signals. The signal is downsampled by $2$ and bandlimited. The shaded areas represent \textit{different} signals: (a) original graph signal, (b) vertex domain downsampling, and (c) graph frequency domain downsampling.}
\label{fig:sampling_gsp}
\end{figure}
\section{Sampling of Graph Signals}
\label{sec:con}
In this section, we introduce the sampling of graph signals in graph frequency domain\cite{Tanaka2017}. Its relationship with the sampling in vertex domain is illustrated in Fig. \ref{fig:sampling_gsp}.
\begin{definition}[Downsampling of graph signals in graph frequency domain]\label{def:GD_spectral} 
Let $\mathbf{L}_0 \in \mathbb{R}^{N \times N}$ and $\mathbf{L}_1 \in \mathbb{R}^{N/M \times N/M}$ be the graph Laplacians for the original graph and for the reduced-size graph\footnote{$M$ is assumed to be a divisor of $N$ for the sake of simplicity.}, respectively, and assume that their eigendecompositions are given by $\mathbf{L}_0 = \mathbf{U}_0 \bm{\Lambda}_0 \mathbf{U}^\top_0$ and $\mathbf{L}_1 = \mathbf{U}_1 \bm{\Lambda}_1 \mathbf{U}^\top_1$. The downsampled graph signal $\bm{f}_{d} \in \mathbb{R}^{N/M}$ in the graph frequency domain from the signal ${\bm{f}} \in \mathbb{R}^{N}$ is defined as 
\begin{equation}
\bm{f}_{d} = \mathbf{U}_1 \mathbf{S}_d \mathbf{U}^\top_0 \bm{f},
\label{S_d}
\end{equation}
 where $\mathbf{S}_d = [
\mathbf{I}_{N/M} \ \mathbf{J}_{N/M} \  \ldots\ \mathbf{I}_{N/M}
\ \mathbf{J}_{N/M}]$, in which $\mathbf{I}_N$ is an $N\times N$ identity matrix and $\mathbf{J}_N$ is an $N\times N$ counter identity matrix.
\end{definition}

\begin{definition}[Upsampling of graph signals in the graph frequency domain]\label{def:GU_spectral} Let $\mathbf{L}_1 \in \mathbb{R}^{N/M \times N/M}$ and $\mathbf{L}_0 \in \mathbb{R}^{N \times N}$ be the graph Laplacians for the original graph and the graph with increased size, respectively. The upsampled graph signal $\bm{f}_{u} \in \mathbb{R}^{N}$ in the graph frequency domain is defined as 
\begin{equation}
\bm{f}_{u} = \mathbf{U}_0 \mathbf{S}_u \mathbf{U}^\top_1 \bm{f}_d,
\label{S_u}
\end{equation}
where $\mathbf{S}_u=\mathbf{S}_d^\top$.
\end{definition}

\begin{figure}[t]
  \centering
	\includegraphics[width=\linewidth]{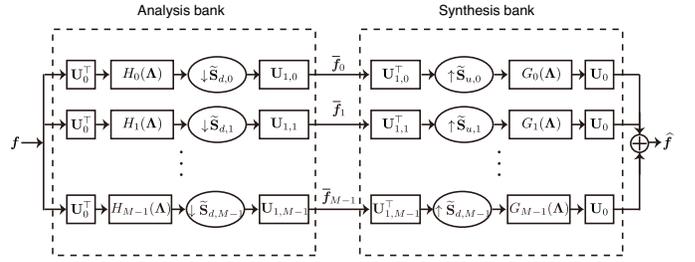}
  \caption{Scheme of $M$-channel spectral graph filter bank.}
  \label{p_st}
\end{figure}

\section{$M$-Channel Spectral Graph Filter Banks}
\label{sec:pro}
\subsection{Framework}
The structure of the proposed $M$-channel spectral graph filter bank is shown in Fig. \ref{p_st}. It is clear that it has a symmetrical structure similar to that of regular signals \cite{Strang1996}.
Here, $H_m({\bf \Lambda})$ and $G_m({\bf \Lambda})$ ($m=0,\ldots,M-1$) are analysis and synthesis filters, respectively. We use slightly modified sampling matrices from \eqref{S_d} and \eqref{S_u}, which are defined as
\begin{equation}
\widetilde{\mathbf{S}}_{d,m}=\begin{cases} 
\begin{bmatrix}
\mathbf{I}_{N/M} &\mathbf{J}_{N/M}& \ldots
\end{bmatrix}&\text{for even $m$},\\
\begin{bmatrix}
\mathbf{I}_{N/M} &-\mathbf{J}_{N/M}& \ldots
\end{bmatrix}
&\text{for odd $m$},\\
\end{cases}
\end{equation}
where $\widetilde{\mathbf{S}}_{u,m} = \widetilde{\mathbf{S}}_{d,m}^\top$ for all $m$.

The $m$th subband signals after the analysis and synthesis transforms are represented respectively as follows.
\begin{equation}
\label{eqn:f_k}
\begin{split}
\widebar{\bm{f}}_m & = \mathbf{U}_{1,m}\widetilde{\mathbf{S}}_{d,m}H_m(\bm{\Lambda})\mathbf{U}_0^\top\bm{f}\\
\widehat{\bm{f}}_m & = \mathbf{U}_0G_m(\bm{\Lambda})\widetilde{\mathbf{S}}_{u,m}\mathbf{U}_{1,m}^\top\widebar{\bm{f}}_m,
\end{split}
\end{equation}
where $\mathbf{U}_{1,m}$ is an arbitrary eigenvector matrix for the $m$th subband.

\subsection{Perfect Reconstruction Condition}
The following theorem describes the perfect reconstruction condition for the proposed scheme.
\begin{theorem}
Assume that $k$ and $p$ are integers within the range of $k\in[0, N]$ and $p\in(-\lfloor kM/N \rfloor, (M-1)-\lfloor kM/N \rfloor]$ ($p\neq 0$), respectively. The $M$-channel spectral graph filter bank defined in the previous subsection is a perfect reconstruction transform, i.e., $\widehat{\bm{f}} = c^2 \bm{f}$, where $c$ is an arbitrary real number, if the graph spectral responses of the filters satisfy the following relationships for all $k$ and $p$.
\begin{equation}
\begin{split}
&\sum_{m=0}^{M-1}G_m(\lambda_k)H_m(\lambda_k)=c^2,
\\
&\sum_{m=0}^{M-1} G_m(\lambda_k)H_m(\lambda_{k+Np/M})=0\text{ for even }p,\\
&\sum_{m=0}^{M-1}(-1)^mG_m(\lambda_k)H_m(\lambda_{(2p+1)N/M-k})=0\text{ for odd }p.
\label{AC_g}
\end{split}
\end{equation}
\end{theorem}
\begin{proof}
The reconstructed signal of the $M$-channel spectral graph filter bank is represented as follows:
\begin{equation}
\begin{split}
\widehat{\bm f}&={\bf U}\left(\sum_{m=0}^{M-1}G_m({\bf \Lambda})\widetilde{\bf S}_{u,m}\widetilde{\bf S}_{d,m}H_m({\bf \Lambda})\right){\bf U}^\top{\bm f}\\
&:={\bf U}\sum_{m=0}^{M-1}{\bf T}_m{\bf U}^\top{\bm f},
\end{split}
\label{r_signal}
\end{equation}
where ${\bf T}_m$ is represented as follows:
\begin{equation}
{\bf T}_m=\begin{bmatrix}
{\bf T}_m(0,0) &\cdots&{\bf T}_m(0,M-1)\\
\vdots&\ddots&\vdots\\
{\bf T}_m(M-1,0) &\cdots&{\bf T}_m(M-1,M-1)\\
\end{bmatrix}.
\end{equation}
Each block in ${\bf T}_m$ is represented as
\begin{equation}
\begin{split}
&{\bf T}_m(k,k+p)=\\
&\begin{cases}
G_m({\bf \Lambda}_k)H_m({\bf \Lambda}_{k+p}) & \text{for any } m\text{ and even }p\\
G_m({\bf \Lambda}_k){\bf J}_{N/M}H_m({\bf \Lambda}_{k+p}) & \text{for even }m \text{ and odd }p.\\
-G_m({\bf \Lambda}_k){\bf J}_{N/M}H_m({\bf \Lambda}_k) & \text{for odd }m \text{ and odd }p.\\
\end{cases}
\end{split}
\end{equation}
If the filter set satisfies \eqref{AC_g}, $\sum_{m=0}^{M-1}{\bf T}_m(k,k+p)$
becomes
\begin{equation}
\begin{split}
&\sum_{m=0}^{M-1}{\bf T}_m(k,k)=\sum_{m=0}^{M-1}G_m({\bf \Lambda}_k)H_m({\bf \Lambda}_k)={\bf I}_{N/M}, \\
&\sum_{m=0}^{M-1}{\bf T}_m(k,k+p)=\\
&\begin{cases}
&\sum_{m=0}^{M-1}G_m({\bf \Lambda}_k)H_m({\bf \Lambda}_{k+p})={\bf 0}_{N/M}\\
 &\ \ \ \ \ \ \ \ \ \ \ \ \ \ \ \ \ \ \ \ \ \ \ \ \ \ \ \ \ \ \ \ \ \ \ \ \ \ \ \ \ 
 \text{for even }p\  (p\neq 0), \\
&\sum_{m=0}^{M-1}(-1)^m G_m({\bf \Lambda}_k){\bf J}_{N/M}H_m({\bf \Lambda}_{k+p})={\bf 0}_{N/M} \\
 &\ \ \ \ \ \ \ \ \ \ \ \ \ \ \ \ \ \ \ \ \ \ \ \ \ \ \ \ \ \ \ \ \ \ \ \ \ \ \ \ \  \text{for odd }p.
\end{cases}
\end{split}
\end{equation}
Then, \eqref{r_signal} becomes
\begin{equation}
\widehat{\bm f}={\bf U}\sum_{m=0}^{M-1}{\bf T}_m{\bf U}^\top{\bm f}={\bf U}{\bf I}_N{\bf U}^\top{\bm f}={\bm f}.
\label{tx}
\end{equation}
\eqref{tx} indicates that the the input signal can perfectly be recovered from the decomposed signals.
\end{proof}

\subsection{Filter Design}
We use filter sets obtained from $M$-channel real-valued linear phase filter banks applied in traditional signal processing by using the design method introduced in \cite{Sakiya2016b}. Let ${\tt H}_m(\omega)$ and ${\tt G}_m(\omega)$ for $m=0,\ldots,M-1$ be arbitrary analysis and synthesis filters in the frequency domain, respectively. We assume that $M$ is even and the filter length is $L=MK$, where $K$ is an arbitrary integer. Furthermore, the even-indexed filters are symmetric, and the odd-indexed filters are anti-symmetric, which are well-accepted assumptions in traditional signal processing, as described in\cite{Gan2003,Tran2000}. The filter ${\tt H}_m(\omega)$ can be represented as 
\begin{equation}
{\tt H}_m(\omega)=
\begin{cases}
e^{-j\frac{L-1}{2}\omega}\Re({\tt H}_m(\omega))& \text{for even $m$},\\
je^{-j\frac{L-1}{2}\omega}\Re({\tt H}_m(\omega))& \text{for odd $m$}.\end{cases}
\label{lpfil}
\end{equation}
Then, the converted filters in the graph frequency domain can be expressed as
\begin{equation}
H_m(\lambda)=
\begin{cases}
e^{j\frac{L-1}{2}\alpha(\lambda)}{\tt H}_m(\alpha(\lambda))& \text{for even $m$},\\
-je^{j\frac{L-1}{2}\alpha(\lambda)}{\tt H}_m(\alpha(\lambda))& \text{for odd $m$},
\end{cases}
\label{conversion}
\end{equation}
where $\alpha(\cdot)$ is a function mapping from $\lambda\in[0,\lambda_\text{max}]$ to $\omega\in[-\pi, \pi]$. The filters have the same filter characteristics as the corresponding classical filter bank in the frequency domain. The filter sets in traditional signal processing satisfy the following perfect reconstruction condition for any $\omega_k\in[0, 2\pi]$ and integer $(-\lfloor M\omega/(2\pi)\rfloor, (M-1)-\lfloor M\omega/(2\pi)\rfloor]$ ($p\neq 0$)\cite{Strang1996}:
\begin{align}
\sum_{m=0}^{M-1}{\tt G}_m^*(\omega_k){\tt H}_m(\omega_k)&=c^2,
\label{eq:classic_PR1}\\
\sum_{m=0}^{M-1}{\tt G}_m^*(\omega_k){\tt H}_m(\omega_k+2\pi p/M)&=0.
\label{eq:classic_PR2}
\end{align}
The following theorem can then be stated.

\begin{theorem}
The spectral graph filters obtained from \eqref{conversion} satisfy the perfect reconstruction conditions of Theorem 1 if $\alpha(\cdot)$ satisfies the following condition:
\begin{equation}
\begin{split}
\alpha(\lambda_{k+pN/M})&=
\begin{cases}
\pi\lambda_{k}/\lambda_{\text{max}}& \text{for } p=0,\\
\alpha(\lambda_{k})+2\pi p/M & \text{for even } p,\\
\end{cases}\\
\alpha(\lambda_{(2p+1)N/M-k})&=
2\pi-\left(\alpha(\lambda_{k})+2\pi p/M\right)\\
&\ \ \ \ \ \ \ \ \ \ \ \ \ \ \ \ \ \ \ \ \ \ \ \text{for odd } p,\\
\label{eq:theo}
\end{split}
\end{equation}
for $k=0,\ldots,N/M-1$ and $p=0,\ldots,M-1$.
\begin{proof}
The proof is described in Appendix.
\end{proof}
\end{theorem}
\begin{figure}[t]
  \centering
	\subfigure[][Minnesota traffic graph (MTG) ]{\includegraphics[width=.4\linewidth]{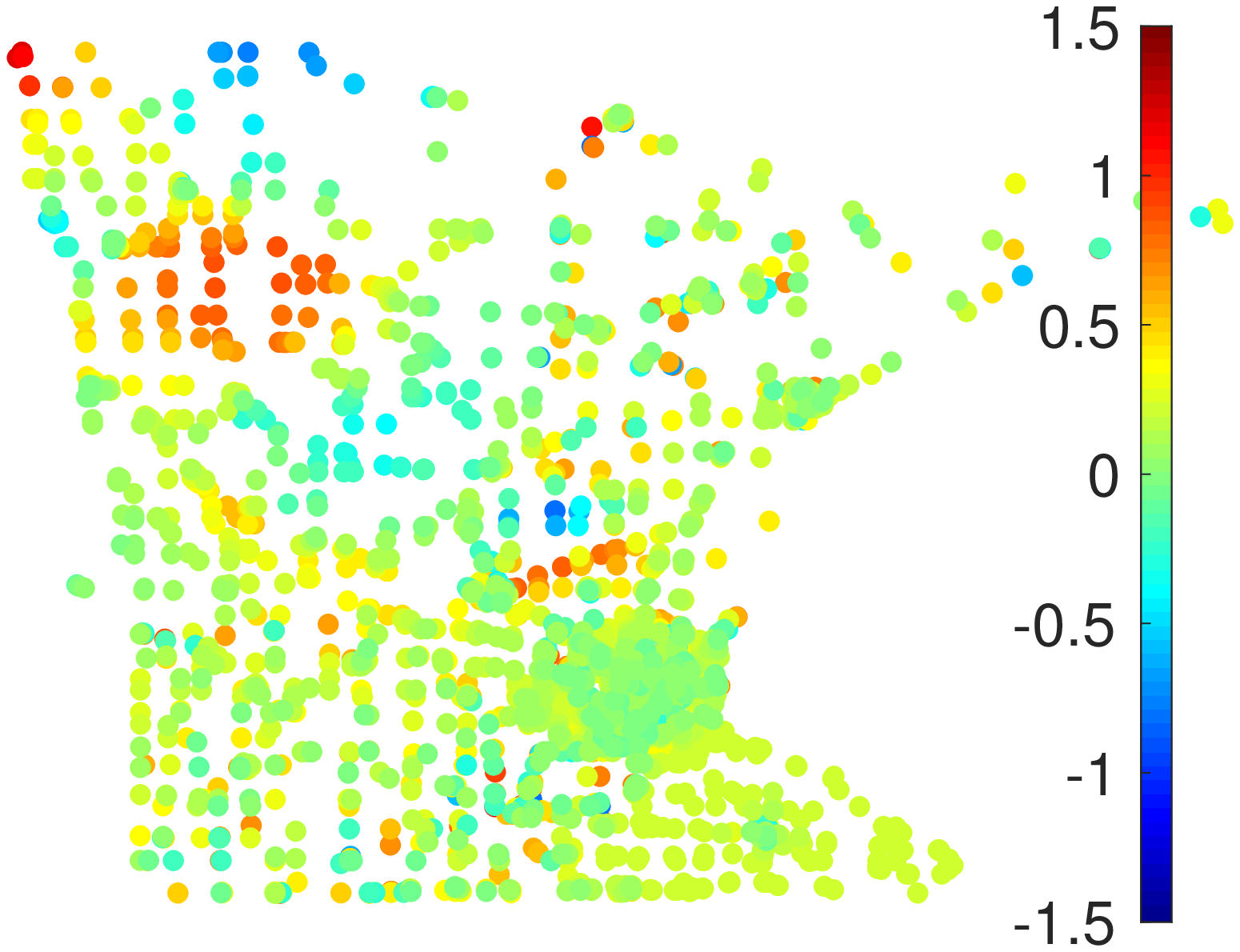}}
	\subfigure[][Spectrum of (a)]{\includegraphics[width=.4\linewidth]{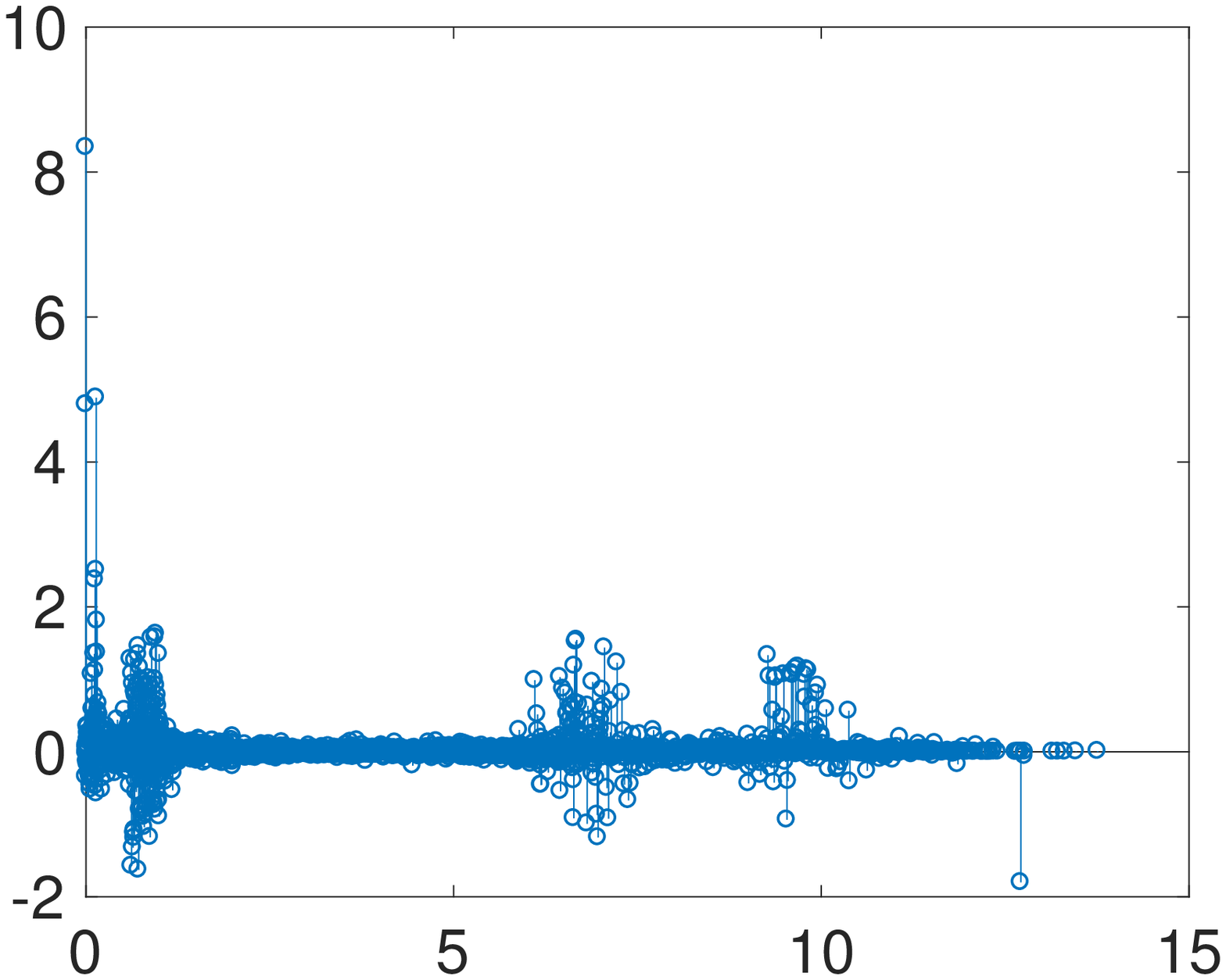}}
	\subfigure[][Random sensor network graph (RSNG)]{\includegraphics[width=.4\linewidth]{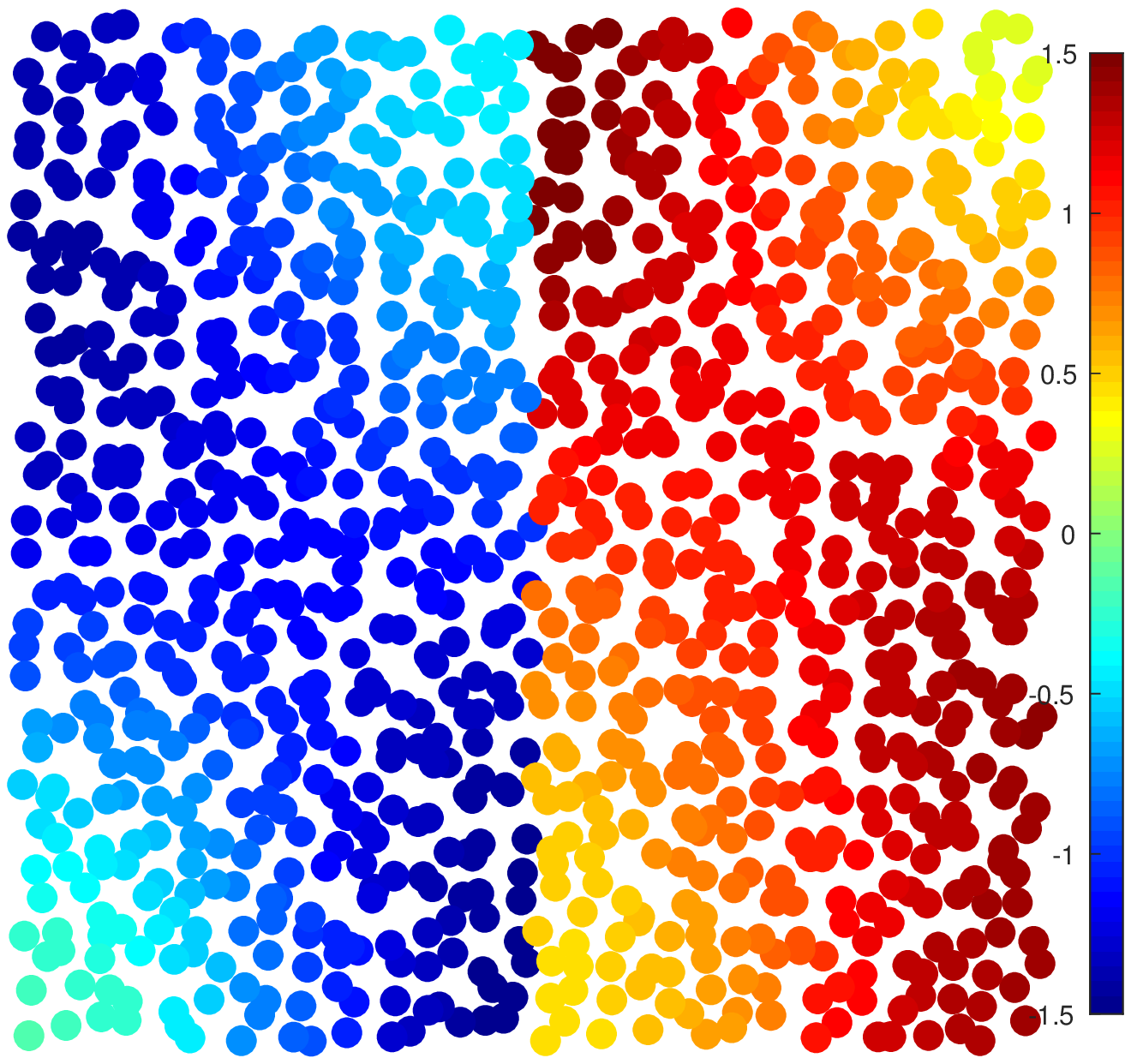}}
	\subfigure[][Spectrum of (c)]{\includegraphics[width=.4\linewidth]{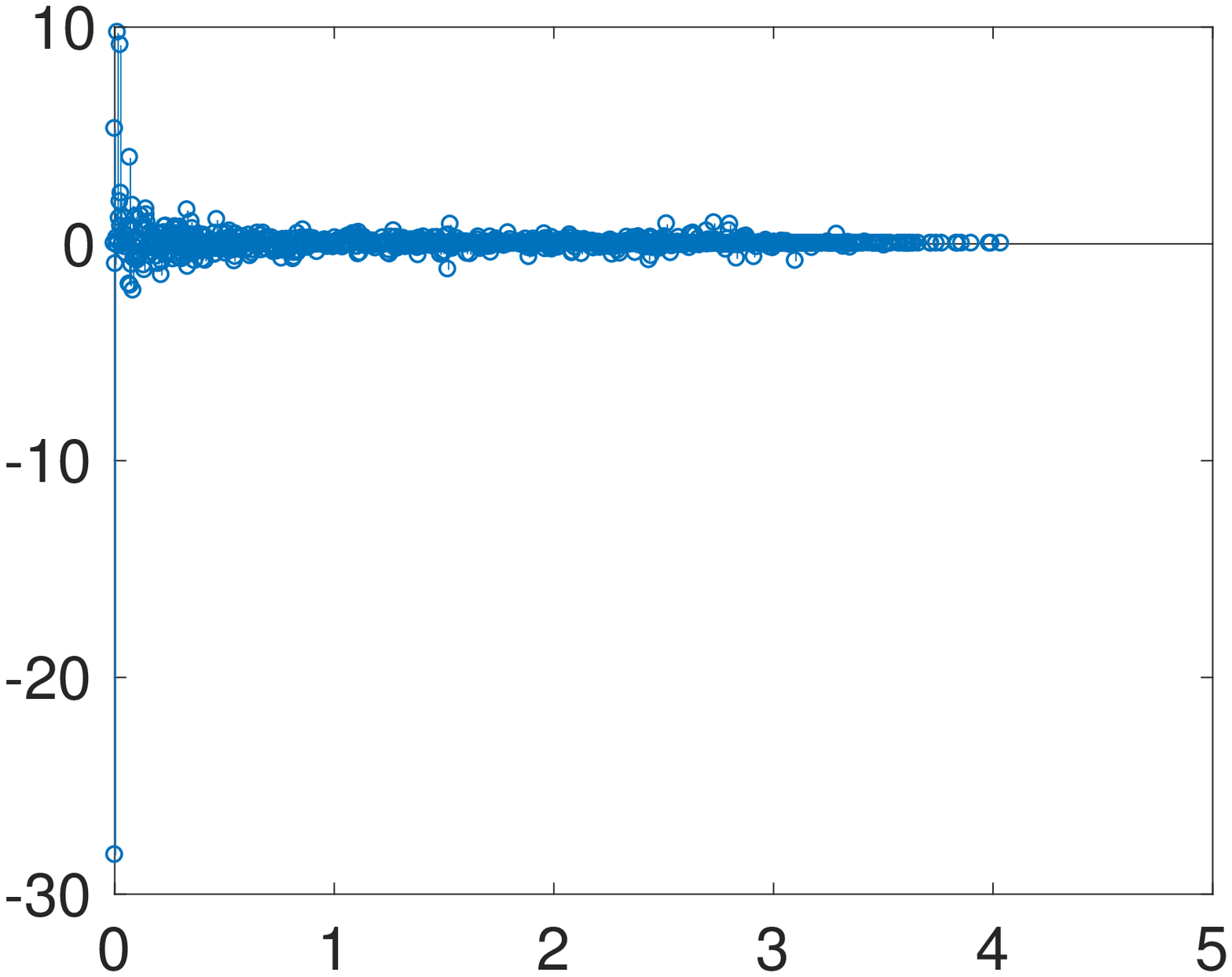}}
  \caption{Original signals.}
  \label{orig}
\end{figure}
\begin{figure}[t]
  \centering
\subfigure[][$4$-GraphSS-L (MTG)]{\includegraphics[width=.4\linewidth]{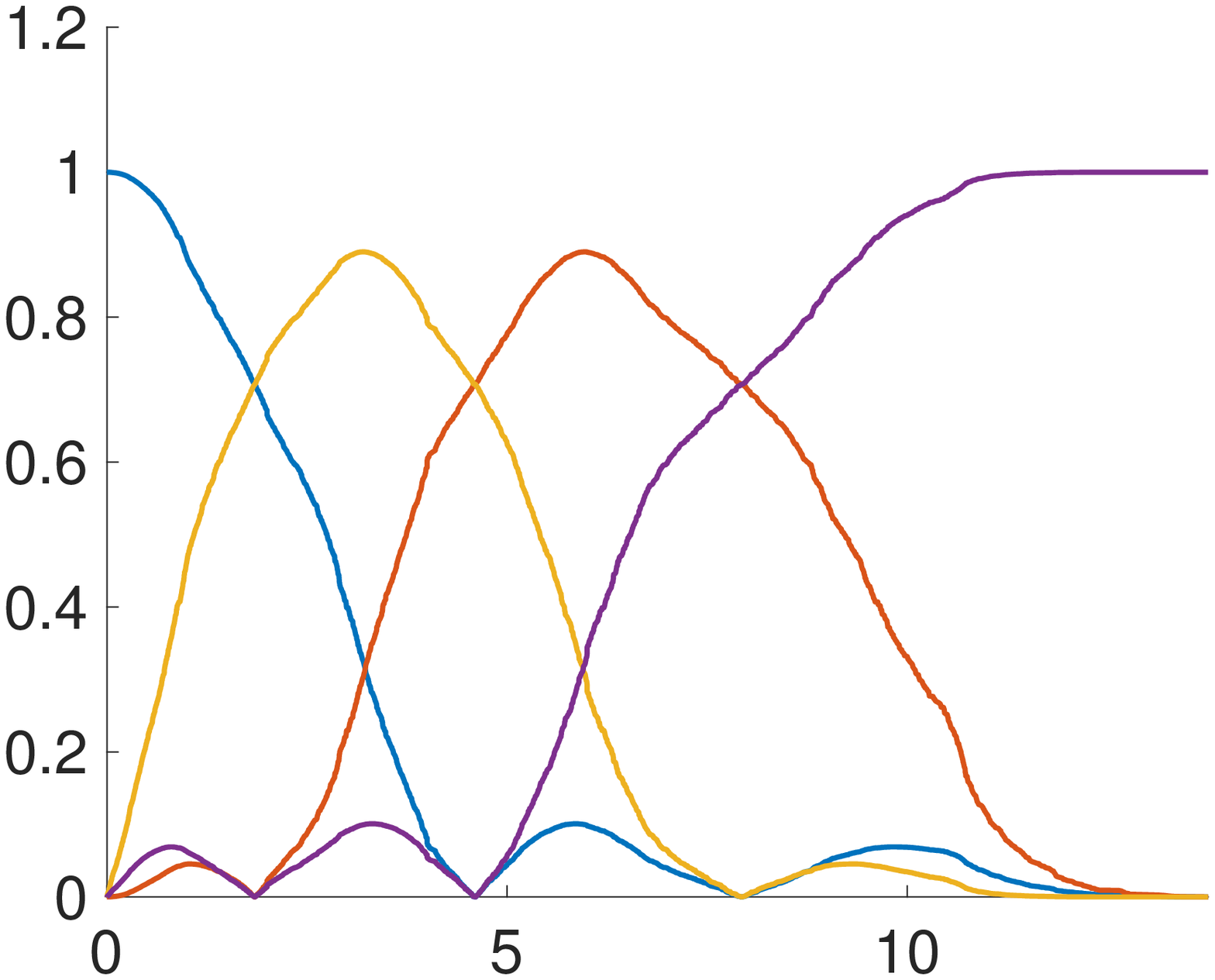}}
\subfigure[][$4$-GraphSS-L (RSNG)]{\includegraphics[width=.4\linewidth]{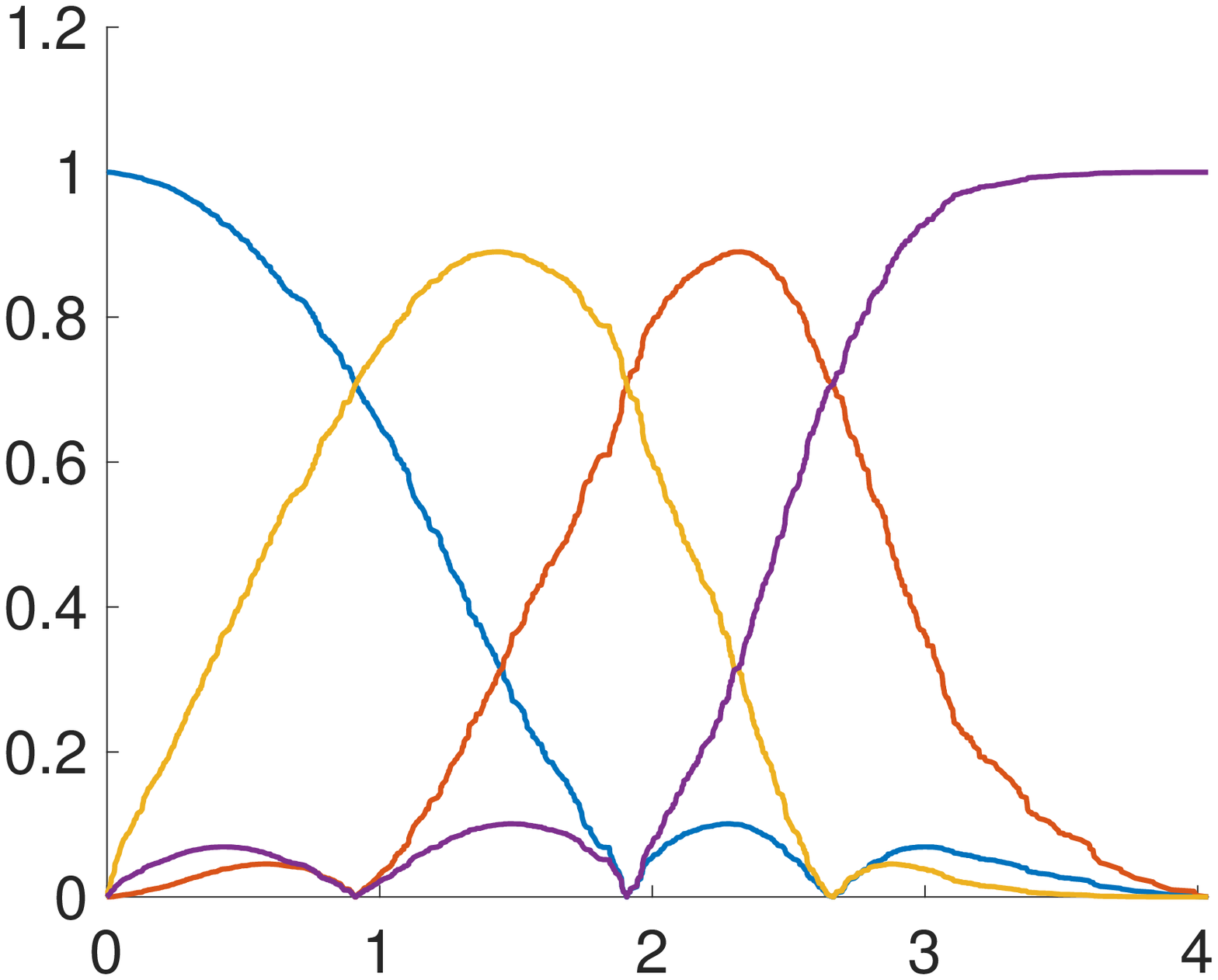}}
\subfigure[][$8$-GraphSS-I (MTG)]{\includegraphics[width=.36\linewidth]{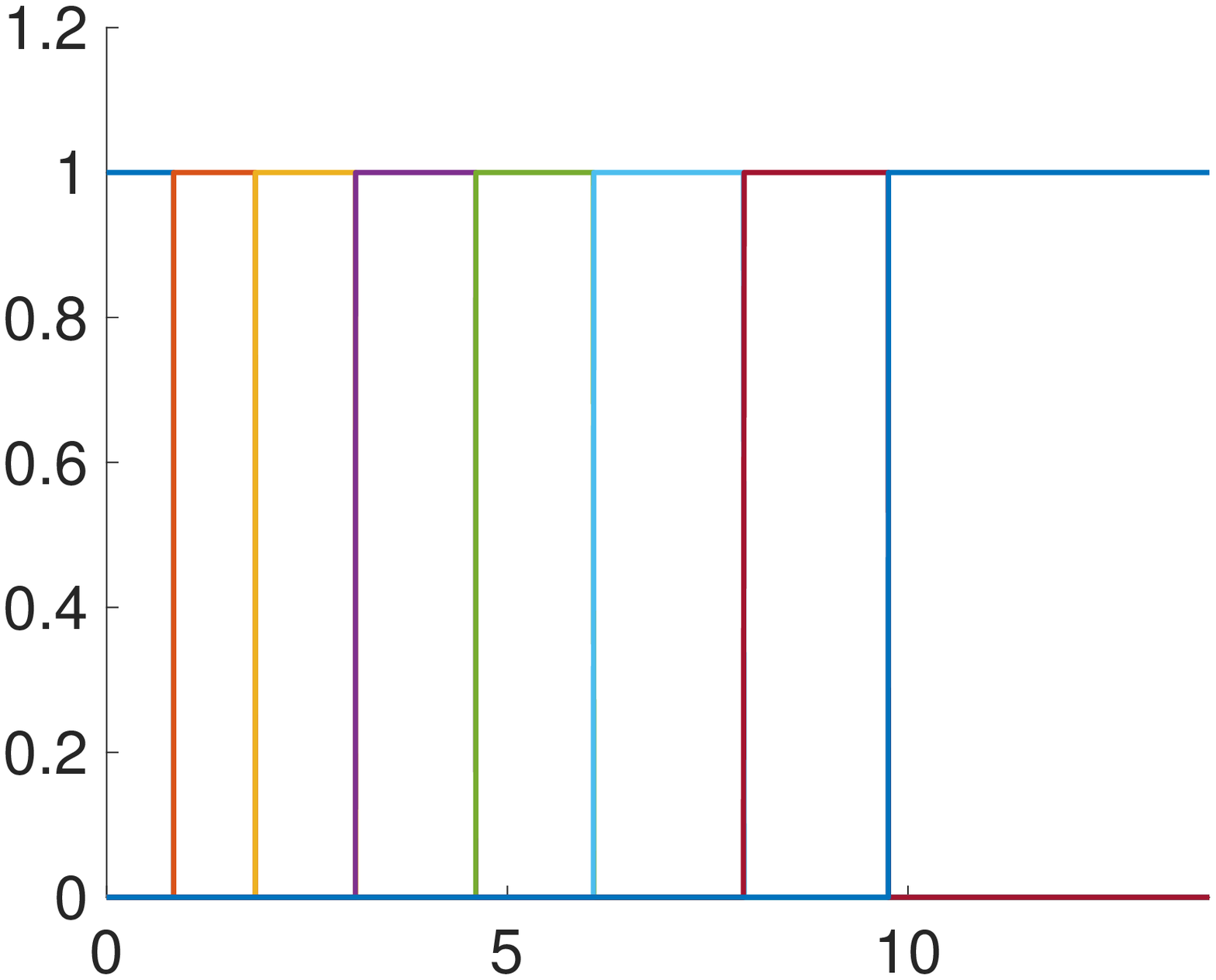}}
\subfigure[][$8$-GraphSS-I (RSNG)]{\includegraphics[width=.4\linewidth]{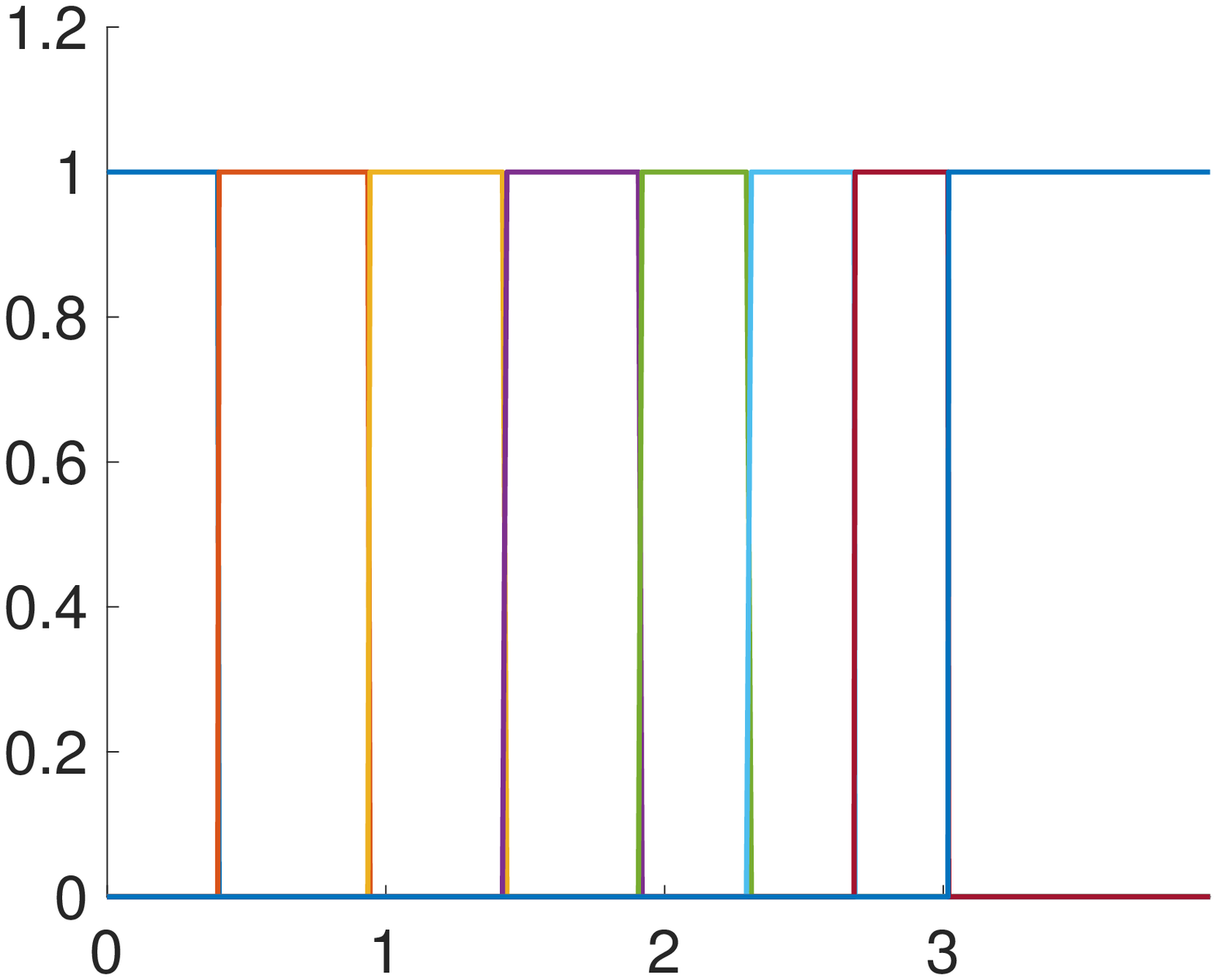}}
  \caption{Filter sets obtained using the eigenvalue distribution of each graph.}
  \label{dct}
\end{figure}

\begin{figure}[t]
  \centering
	\subfigure[][Minnesota traffic graph]{\includegraphics[width=.82\linewidth]{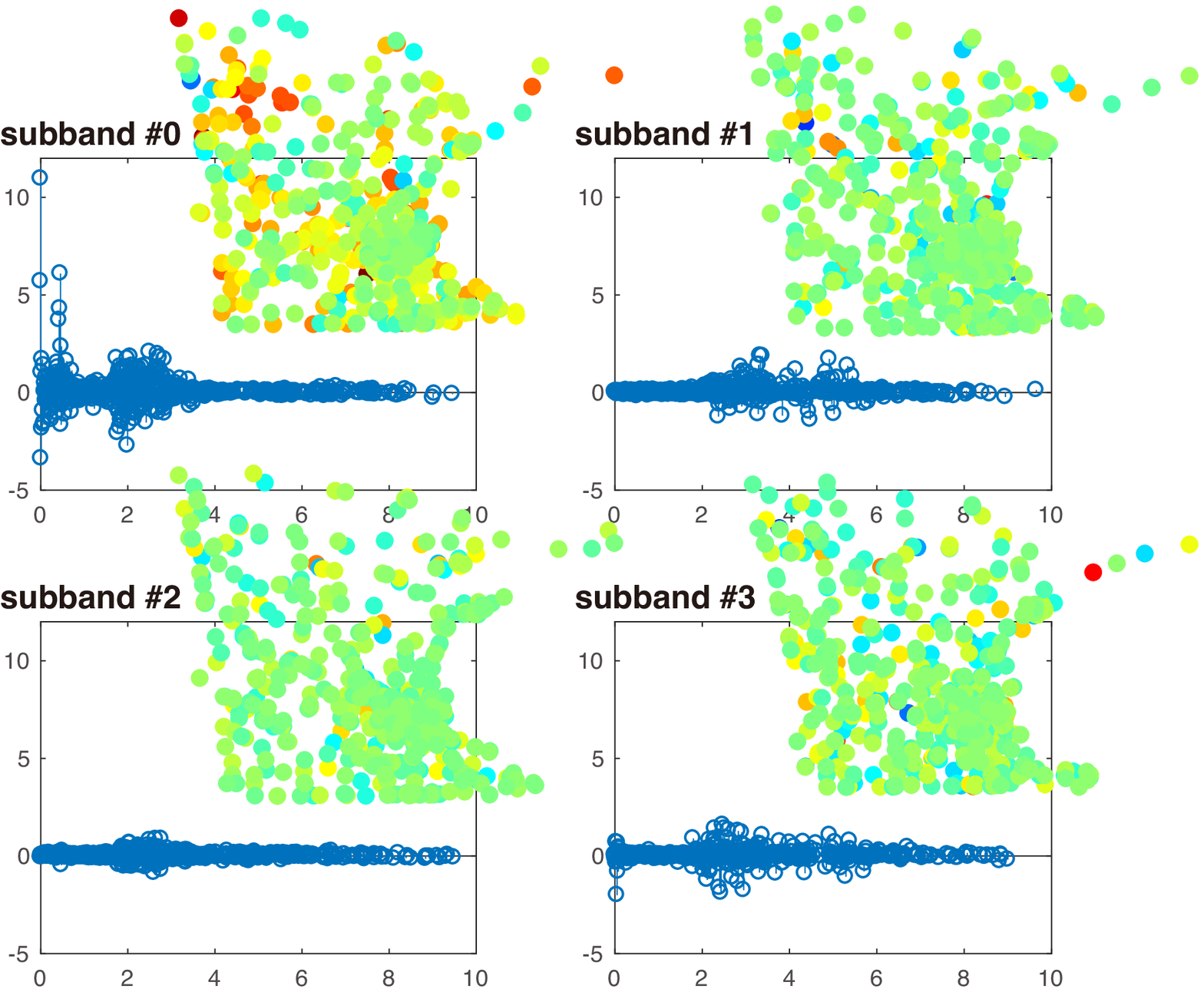}}
	\subfigure[][Random sensor network graph]{\includegraphics[width=.82\linewidth]{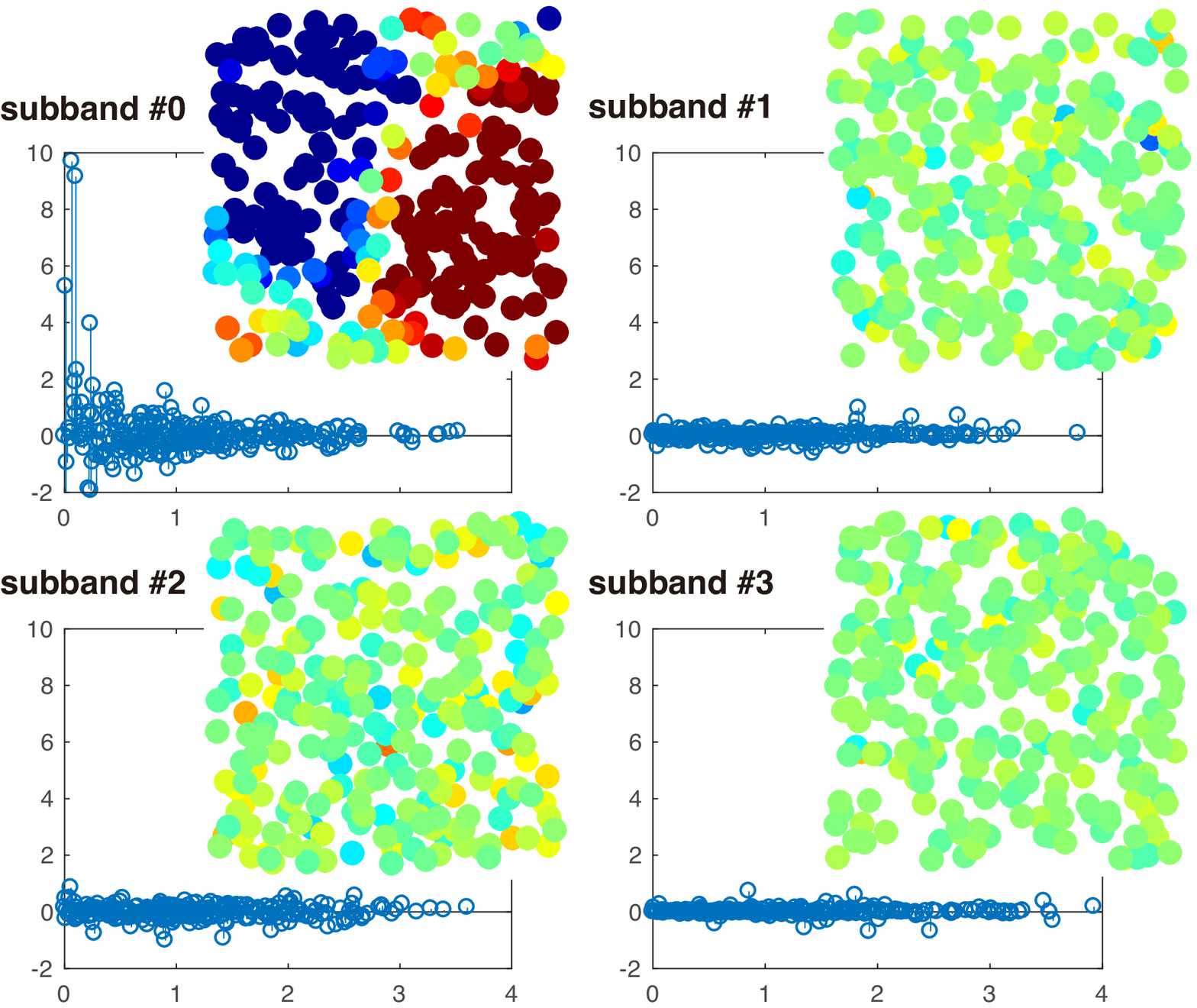}}
  \caption{Decomposition results in the vertex and graph frequency domains.}
  \label{8-ch}
\end{figure}

\section{Experimental Results}
\label{sec:exp}

Hereafter, we abbreviate the proposed filter bank as $M$-\textit{graphSS}, where SS refers to the spectral sampling, and the filter bank using vertex sampling \cite{Narang2012, Narang2013} is abbreviated as \textit{graphVS}, where VS refers to vertex sampling.

\subsection{Signal Decomposition}
The graph signals shown in Fig. \ref{orig} are decomposed through $4$-GraphSS with lapped orthogonal transform (LOT) based filters (denoted as $4$-GraphSS-L) \cite{Malvar1989}, as shown in Figs. \ref{dct} (a) and (b). The filters are orthogonal, and therefore, we can use the same filter sets on the synthesis side.  
The decomposed results are shown in Fig. \ref{8-ch}. The vertices maintained in each subband are selected according to the vertex indices. We can see that the decomposed signals contain the characteristics of the spectrum in the original signal, and the transform divides the original signal into different frequency ranges.

\begin{table}[t]
  \centering 
  \caption{Denoised Results: SNR (dB)}\label{tab:min_denoise}
  \begin{tabular}{c|c|c|c|c|c|c}
\hline
$\sigma$  &1& 1/2 & 1/4 & 1/8 & 1/16 & 1/32\\\hline\hline
\multicolumn{7}{c}{\textit{Random Sensor Network Graph}}\\\hline
GFT	&	{\bf 6.44}	&	9.40	&	13.92	&	19.22	&	25.06	&	30.86	\\
GraphBior	&	4.45	&	8.46	&	13.57	&	19.02	&	24.97	&	30.82	\\
GraphVS-I&	4.43	&	8.54	&	13.58	&	19.00	&	24.98	&	30.82	\\
$2$-GraphSS-I	&	3.36	&	9.13	&	14.44	&	19.66	&	25.29	&	30.95	\\
$8$-GraphSS-I	&	4.71	&	{\bf 10.13}	&	{\bf 15.60}	&	{\bf 20.39}	&	{\bf 25.68}	&	{\bf 31.11}	\\\hline
noisy	&	0.69	&	6.78	&	12.83	&	18.76	&	24.88	&	30.80	\\\hline
\multicolumn{7}{c}{\textit{Minnesota Traffic Graph}}\\\hline
GFT	&	-1.71	&	7.63	&	11.80	&	16.92	&	22.64	&	28.53	\\
GraphBior	&	2.10	&	6.13	&	11.08	&	16.58	&	22.51	&	28.48	\\
GraphVS-I&	2.24	&	6.21	&	11.10	&	16.60	&	22.51	&	28.48	\\
$2$-GraphSS-I	&	0.19	&	5.66	&	11.12	&	16.72	&	22.58	&	28.51	\\
$8$-GraphSS-I	&	{\bf 3.47	}&	{\bf 7.94	}&	{\bf 12.36}	&	{\bf 17.27}	&	{\bf 22.79}	&	{\bf 28.58}	\\\hline
noisy	&	-1.71	&	4.35	&	10.32	&	16.32	&	22.42	&	28.45	\\\hline
  \end{tabular}
  \label{table:image}
\end{table}

\subsection{Denoising}
To evaluate the performance of the proposed spectral graph filter bank, the denoising is performed. 
We used $8$-GraphSS with ideal filters (denoted as $8$-GraphSS-I), as shown in Figs. 4 (c) and (d), which is compared with the GFT, graphBior \cite{Narang2013}, $2$-GraphVS with an ideal filter \cite{Narang2012} (denoted as GraphVS-I),  and $2$-GraphSS-I \cite{Watana2018}.
For a fair comparison, two-channel GFBs use $7$-level octave decompositions, i.e., they also have eight subbands. 
The input signal is corrupted by additive white Gaussian noise.
For denoising, the coefficients in the $Y$th subband are thresholded \cite{Chang2000a} as follows:
\begin{equation}
\begin{split}
T=\sigma^2/\sqrt{\max(\sigma_Y^2-\sigma^2,0)},
\end{split}
\end{equation}
where $\sigma$ is the standard deviation of the noise, and $\sigma_Y$ is the variance of coefficients in the $Y$th subband.
 The input signals are shown in Fig. \ref{orig}, and the denoising results are shown in Table I. The proposed method shows better results than conventional methods in most cases.

\section{Conclusion}
\label{sec:conc}
In this paper, $M$-channel spectral graph filter banks using sampling in the graph frequency domain are proposed. They can be applied to any graph signals and can use any variation operators. We demonstrated perfect reconstruction conditions and filter design methods. The proposed graph filter bank can decompose the signals while maintaining the spectrum of the original graph signal, and in an experiment on denoising showed better results than a conventional transform with maximum decimation.

\section*{Appendix}
We assume that $\omega_k=\alpha(\lambda_k)$ for $k=0,\ldots,N/M-1$. 
If $m$ is even and $k<N/M$, from \eqref{conversion} and \eqref{eq:theo}, $G_m(\lambda_k)H_m(\lambda_k)$ can be rewritten as
\begin{equation}
\begin{split}
&G_m(\lambda_k)H_m(\lambda_k)\\
&=e^{j\frac{L-1}{2}\alpha(\lambda_k)}{\tt G}_m(\alpha(\lambda_k))e^{j\frac{L-1}{2}\alpha(\lambda_k)}{\tt H}_m(\alpha(\lambda_k))\\
&=e^{-j\frac{L-1}{2}\alpha(\lambda_k)}e^{j\frac{L-1}{2}\alpha(\lambda_k)}\Re({\tt G}_m(\alpha(\lambda_k)))\\
&\phantom{=\ }\times e^{j\frac{L-1}{2}\alpha(\lambda_k)}e^{-j\frac{L-1}{2}\alpha(\lambda_k)}\Re({\tt H}_m(\alpha(\lambda_k)))\\
&={\tt G}_m^*(\omega_k){\tt H}_m(\omega_k).
\end{split}
\label{pro_1}
\end{equation}

If $m$ and $p$ are even, from \eqref{lpfil}, \eqref{conversion} and \eqref{eq:theo}, $G_m(\lambda_k)H_m(\lambda_{(2p+1)N/M-k})$ can be rewritten as
\begin{equation}
\begin{split}
&G_m(\lambda_k)H_m(\lambda_{(2p+1)N/M-k})\\
&=-je^{j\frac{L-1}{2}\alpha(\lambda_k)}{\tt G}_m(\alpha(\lambda_k))(-j)e^{j\frac{L-1}{2}(2\pi-(\alpha(\lambda_k)+2\pi p/M))}\\
&\phantom{=\ }\times {\tt H}_m(2\pi-(\alpha(\lambda_k)+2\pi p/M))\\
&=-je^{j\frac{L-1}{2}\omega_k}{\tt G}_m(\omega_k)(-j)e^{j\frac{L-1}{2}(2\pi-(\omega_k+2\pi p/M))}\\
&\phantom{=\ }\times {\tt H}_m^*(\omega_k+2\pi p/M)\\
&=-je^{j\frac{L-1}{2}\omega_k}je^{-j\frac{L-1}{2}\omega_k}\Re({\tt G}_m(\omega_k))\\
&\phantom{=\ }\times (-j)e^{j\frac{L-1}{2}(2\pi-(\omega_k+{2\pi p/M}))}(-j)e^{j\frac{L-1}{2}(\omega_k+{2\pi p/M})}\\
&\phantom{=\ }\times \Re({\tt H}_m(\omega_k+{2\pi p/M}))\\
&=e^{j\frac{L-1}{2}({2\pi p/M})}{\tt G}_m^*(\omega_k){\tt H}_m(\omega_k+{2\pi p/M})\\
\end{split}
\label{pro_1}
\end{equation}
for any $k$. By using similar approach, the following equations are derived:
\begin{equation}
G_m(\lambda_k)H_m(\lambda_k)={\tt G}_m^*(\omega_k){\tt H}_m(\omega_k)
\end{equation}
for any $k$ and $m$,
\begin{equation}
G_m(\lambda_k)H_m(\lambda_{k+pN/M})=e^{j\frac{L-1}{M}p\pi}{\tt G}_m^*(\omega_k){\tt H}_m(\omega_k+{2\pi p/M})
\end{equation}
for any $k$ and $m$, and even $p$, and
\begin{align}
&G_m(\lambda_k)H_m(\lambda_{(2p+1)N/M-k})=\notag\\
&\begin{cases}
e^{j\frac{L-1}{2}({2\pi p/M})}{\tt G}_m^*(\omega_k){\tt H}_m(\omega_k+{2\pi p/M})& \text{if } m \text{ is even},\\
-e^{j\frac{L-1}{2}({2\pi p/M})}{\tt G}_m^*(\omega_k){\tt H}_m(\omega_k+{2\pi p/M})& \text{if } m \text{ is odd},
\end{cases}
\end{align}
for any $k$ and $m$, and odd $p$.
Then, the following relationships are satisfied:
\begin{align}
&\sum_{m=0}^{M-1}G_m(\lambda_k)H_m(\lambda_k)=\sum_{m=0}^{M-1}{\tt G}_m^*(\omega_k){\tt H}_m(\omega_k)=c^2,\\
&\sum_{m=0}^{M-1}G_m(\lambda_k)H_m(\lambda_{k+2pN/M})\notag\\
&=e^{j\frac{L-1}{M}p\pi}\sum_{m=0}^{M-1}{\tt G}_m^*(\omega_k){\tt H}_m(\omega_k+{2\pi p/M})=0,\\
&\sum_{m=0}^{M-1}(-1)^mG_m(\lambda_k)H_m(\lambda_{(2p+1)N/M-k})\notag\\
&
=e^{j\frac{L-1}{2}({2\pi p/M})}\sum_{m=0}^{M-1}{\tt G}_m^*(\omega_k){\tt H}_m(\omega_k+{2\pi p/M})=0
\end{align}
They coincide with \eqref{AC_g}.


\end{document}